\documentclass[english]{cccconf}
\usepackage[comma,numbers,square,sort&compress]{natbib}
\usepackage{epstopdf}

\usepackage{amsmath}
\usepackage{amsthm}
\usepackage{listings}
\usepackage{color}
\usepackage{graphicx}
\usepackage{tikz}
\usepackage{pgfplots}
\usepackage{pgfplotstable}
\usepackage{indentfirst}
\usepackage{amssymb}
\usepackage{hyperref}

\newtheorem{definition}{Definition}[section]
\newtheorem{remark}{Remark}[section]
\newtheorem{theorem}{Theorem}[section]
\newtheorem{lemma}{Lemma}[section]

\newtheorem{algorithm}{Algorithm}[section]

\begin{document}

\title{Security Protection in Cooperative Control of Multi-agent Systems}

\author{Shuai Sun\aref{amss},
        Yilin Mo\aref{amss}}

% Note: the first argument in the \affiliation command is optional.
% It defines a label for the affiliation which can be used in the \aref
% command. If there is only one affiliation for all authors, then the
% optional argument in the \affiliation command should be suppressed,
% and the \aref command should also be removed after each author in
% \author command, in this case the affiliation will not be numbered.

\affiliation[amss]{
	Department of Automation and BNRist, Tsinghua University, Beijing 100084, P.~R.~China
	\email{suns19@mails.tsinghua.edu.cn, ylmo@tsinghua.edu.cn}}
%\affiliation[hit]{Center for Intelligent and Networked Systems, Department of Automation, BNRist, Tsinghua University, Beijing 100084, P.~R.~China
%        \email{ylmo@tsinghua.edu.cn}}

\maketitle

\begin{abstract}
		Due to the wide application of average consensus algorithm, its security and privacy problems have attracted great attention. In this paper, we consider the system threatened by a set of  unknown agents that are both ``malicious'' and ``curious'', who add additional input signals to the system in order to perturb the final consensus value or prevent consensus, and try to infer the initial state of other agents.
	At the same time, we design a privacy-preserving average consensus algorithm equipped with an attack detector with a time-varying exponentially decreasing threshold for every benign agent, which can guarantee the initial state privacy of every benign agent, under mild conditions.
	The attack detector will trigger an alarm if it detects the presence of malicious attackers.
	An upper bound of false alarm rate in the absence of malicious attackers and the necessary and sufficient condition for there is no undetectable input by the attack detector in the system are given.
	Specifically, we show that under this condition, the system can achieve asymptotic consensus almost surely when no alarm is triggered from beginning to end, and an upper bound of convergence rate and
	some quantitative estimates about the error of final consensus value are given.
	Finally, numerical case is used to illustrate the effectiveness of some theoretical results.
\end{abstract}


% Please remove or comment out the following line if the footnote is not necessary
%\footnotetext{This work is supported by National Natural Science
%Foundation (NNSF) of China under Grant 00000000.}

\section{Introduction}

	Multi-agent systems have attracted widespread attention in recent years because of their better flexibility, good scalability, and excellent computing performance\cite{Bullo F}. Consensus is one of the most common tasks in multi-agent systems \cite{LiuQ} with applications in distributed estimation and optimization \cite{asqualetti}
	%\cite{Mateos}
	, sensor fusion \cite{Xiao}, distributed energy management \cite{Zhao} and sensing scheduling \cite{He}, and time synchronization \cite{Schenato}
	%\cite{Carli}\cite{He2011} 
	and so on.
	Under the traditional average consensus algorithm of discrete-time, at each time step, each agent updates its state value to be a weighted average of its own previous state value and those of its neighbors. 
	Since there is no fusion center that can monitor the behavior of all agents at any time step, systems are very vulnerable to internal and external attacks \cite{Michiardi2002}. Attackers can cause a series of serious problems, such as system security and internal privacy issues.
	
	``Malicious" and ``curious" attackers are two common attackers. 
	``Malicious'' attackers do not follow the average consensus algorithm but add additional input signals to the system in order to perturb the final consensus or prevent consensus. ``Curious'' attackers try to infer the initial states of other agents based on the update rule of the average consensus algorithm. This is extremely unfavorable in a privacy-sensitive situation. 
	
	In order to address the urgent need for security and privacy,  a number of security and privacy protection methods related to the average consensus algorithm have been proposed. In order to ensure the security of consensus, in \cite{Sundaram2008_1}\cite{Sundaram2008_2}, Sundaram and Hadjicostis used the method of parity space for fault detection to show the resilience of linear consensus network from the perspective of network topology. Pasqualetti et al discussed the relationship between consensus computation in unreliable networks and fault detection and isolation problem for linear systems and gave some attack detection and identification algorithms based unknown input observer method\cite{Bullo F}.
	On the other hand, to protect privacy, Huang et al. proposed an average consensus algorithm that adds Laplacian noise with exponential decay characteristics to the calculation, but the resulting convergence value is a random value \cite{Huang2012}. 
	In \cite{Man 2013}, Manitara and Hadjicostis proposed a privacy-preserving average consensus protocol and showed that the privacy of the initial state can be guaranteed when the network topology satisfies certain conditions, but they did not provide quantitative results on how good the initial state can be estimated. Mo and Murray proposed a privacy-preserving average consensus algorithm and proved that the initial state privacy of every benign agent can be effectively protected, under mild situations \cite{Mo 17 }. In \cite{Wang 2019}, Wang proposed a privacy-preserving protocol in which the state of every agent is randomly decomposed into two substates, such that the mean remains the same but only one of them is revealed to other neighboring agents. Hadjicostis and Dom\'inguez-Garc\'ia addressed the problem of privacy-preserving asymptotic average consensus in the presence of curious agents by using homomorphic encryption\cite{Had2020}.
	
	However, the security and privacy problems will become more difficult to deal with when attackers are both ``malicious'' and ``curious''. \cite{Ruan2017} proposed a homomorphic cryptography-based approach with high computational complexity, which can guarantee privacy and security in decentralized average consensus. Nevertheless, the security problem considered in \cite{Ruan2017} is the security of communication rather than against malicious attackers.
	In \cite{LiuQ}, Liu et al proposed a privacy-preserving average consensus algorithm equipped with a malicious attack detector by using the method of state estimation and used the reachable set to characterize the maximum disturbance that the attackers can introduce to the system.
	
	In this paper, we consider the case where the system is threatened by a set of unknown attackers that are both ``malicious" and ``curious". The main differences between this paper and the reference \cite{LiuQ} are as follows. (1). We use an orthogonal projection matrix of the observation matrix of the system to construct the residual vectors, which is used to design an attack detector, while \cite{LiuQ} used the method of state estimation.
	(2). We give the necessary and sufficient condition for there is no undetectable input in the system. Further, under this condition, we show that the system can achieve asymptotic consensus, and give an upper bound of convergence rate, 
	while the corresponding content is missing in \cite{LiuQ}.
	(3). We give some quantitative results about the estimate of the error of final consensus value from the perspective of theoretical analysis when asymptotic consensus is reached. However, \cite{LiuQ} characterized the maximum disturbance that the attackers can introduce to the system by using the method of ellipsoid approximation of reachable set, and the estimate error region may be unbounded.
	
	The main contributions of this paper are as follows.
	Based on the privacy-preserving consensus algorithm proposed in \cite{Mo 17 }, we design a privacy-preserving average consensus algorithm equipped with an attack detector with a time-varying exponentially decreasing threshold for every benign agent, which can guarantee the initial state privacy of every benign agent, under mild conditions.
	The detector will trigger an alarm if it detects the presence of malicious attackers.
	An upper bound of false alarm rate in the absence of malicious attacker and the necessary and sufficient condition for there is no undetectable input in the system are given.
	Under this condition, we show that the system can achieve asymptotic consensus almost surely when no alarm is triggered from beginning to end and give an upper bound of convergence rate and
	some quantitative results about the estimate of the error of final consensus value from the perspective of theoretical analysis.
	
	The rest of this paper is organized as follows: Section 2 briefly reviews the average consensus algorithm and introduces two kinds of attack models. 
	Sections 3 and 4 give the relevant results of privacy protection  against curious attackers and security protection against malicious attackers, respectively, including the specific definition of concepts and detailed proof of theorems. Section 5 gives numerical case to illustrate the effectiveness of some theoretical results and Section 6 concludes this paper.
	
	\textbf{Notations:} $ \mathbb{N} $ is the set of all non-negative integers. $\mathbb{R}^{n}$ is the set of $n\times 1$ real vectors. $\mathbb{R}^{n\times m}$ is the set of $n\times m$ real matrices. ${\rm tr}{M}$ is the trace of square matrix $M$. $\mathbf{1}$  is an all one vector of proper dimension. $\mathbf{0}$  is an all zero matrix of proper dimension. $\|v\|$ indicates the 2-norm of the vector $v$, while $\|M\|$ is the induced 2-norm of the matrix $M$. $X^+$ is the Moore--Penrose pseudoinverse of the matrix $X$. $  \{a(k)\}_{k=0}^{n}  $ stands for the finite set $\{ a(0), a(1), \cdots, a(n)\}$ and $  \{a(k)\}_{k=0}^{\infty}  $ stands for the infinite set $ \{a(0), a(1), \cdots \}$.

\section{Problem Formulation}
	\subsection{Average Consensus}
	In this subsection, we briefly introduce the average consensus algorithm.
	
	Consider a network composed by $n$ agents as an undirected connected graph $G=(V,E)$, where $V=\{1,2,\cdots,n\}$ is the set of agents, and $E\subseteq V\times V$ represents the communication relationship among the agents. An edge between $i$ and $j$, denoted by $(i,j)\in E$, implies that $i$ and $j$ can communicate with each other. The set of neighbors of $i$ is denoted by $\mathcal{N}_i=\{j\in V:(i,j)\in E, j \neq i \}$.
	
	Suppose that each agent $i\in V$ has an initial state $x_i(0)$. At any time $k$, agent $i$ first broadcasts its state to all of its neighbors and then updates its own state in the following linear combination manner:
	\begin{equation}
	\label{eq1}
	x_i(k+1)= a_{ii}x_i(k) + \sum_{j\in \mathcal{N}_i}a_{ij}x_j(k).
	\end{equation}
	where $a_{ij}\ne 0$ if and only if $i$ and $j$ are neighbors. Define $x(k) \triangleq  [x_1(k), x_2(k), \cdots, x_n(k)]^\top \in \mathbb{R}^n$ and $A\triangleq [a_{ij}]\in \mathbb{R}^{n\times n}$, where $ A $ is called \textit{weight matrix}. The state updating rule can be written in the following matrix form:
	\begin{equation}
	\label{eq2}
	x(k+1)=Ax(k).
	\end{equation}
	We say the agents reach a consensus if $
	\lim_{k\rightarrow \infty}x(k)= \gamma \mathbf{1}_{n \times 1}$,
	where $\gamma$ is an arbitrary scalar constant. If $\gamma=\frac{1}{n}\sum_{i=1}^nx_i(0)$, then we say the average consensus is reached.
	
	Assume the eigenvalues of $A$ are arranged as $\lambda_1\ge \lambda_2 \ge \cdots \ge \lambda_n$. It is well known that the necessary and sufficient conditions for average consensus are as follows:
	
	(A1) $\lambda_1=1$, and $|\lambda_i|<1$, $i=2,\cdots,n$;
	
	(A2) $A\mathbf{1}_{n \times 1}=\mathbf{1}_{n \times 1}$.
	
	In the rest of this paper, assume that $ A $ is \textbf{symmetric} and satisfies Assumption (A1) and (A2) above.

\subsection{Attack Models}
In this subsection, we introduce two kinds of attack models.

\textit{Malicious Attack:} Some agents intend to disrupt the average consensus or prevent consensus by adding arbitrary input signals instead of following the updating rule (\ref{eq1}) of average consensus algorithm, i.e.,
\begin{equation}
\label{eq4}
x_i(k+1)= a_{ii}x_i(k) +  \sum_{j\in \mathcal{N}_i}a_{ij}x_j(k)+u_i(k),
\end{equation}
where $u_i(k)\neq 0$ is the attack signal added by agent $i$ at time step $k$. Agent $ i $ is said to be a \textit{malicious} attacker if $ u_i(k) $ is nonzero for at least one time step $ k $, $ k \in \mathbb{N} $. The model for malicious attackers considered here is quite general, and the attack signal at every time step can be an arbitrarily determinant value.

This kind of malicious attackers can potentially either prevent benign agents, who follow the standard update rule  \eqref{eq1} of
the average consensus algorithm, from reaching a consensus or manipulate the final consensus value to be arbitrary.

\textit{Curious Attack:} Some agents intend to infer the initial states of other agents, which may not be desirable when the initial state privacy is of concern. Such agents are called \textit{curious} attackers.

In this paper, we deal with a set of unknown agents that are both ``malicious'' and ``curious''. We assume that the set of these unknown agents that are both ``malicious'' and ``curious'' is $ \{i_1,\cdots,i_p\} $. Meanwhile, assume that each agent knows the weight matrix $ A $ defined in the previous subsection.

\section{Privacy Protection Against Curious Attackers}
In this section, we address the problem of curious attackers inferring other benign agents' initial states.
\subsection{Privacy Preserving Consensus Algorithm}

In order to protect each benign agent's privacy, we adopt the privacy-preserving algorithm proposed in \cite{Mo 17 }. 
For the sake of completeness, we briefly describe the algorithm as follow.

\begin{algorithm}\label{algorithm}
	Let $v_i(k) (i=1,2,\cdots,n; k=0,1,\cdots)$ be standard normal distributed random variables, which are independent across $i$ and $k$. Denote $v(k) \stackrel{\bigtriangleup}{=}[v_1(k),v_2(k),\cdots,v_n(k)]^\top $. Based on $v(k)$ we can construct the following noisy signals
	\begin{equation}
	\label{eq5}
	w(k)=\left\{
	\begin{array}{ll}
	v(0), & \hbox{if} ~~k=0; \\
	\varphi^kv(k)-\varphi^{k-1}v(k-1), & \hbox{otherwise};
	\end{array}
	\right.
	\end{equation}
	where $0 < \varphi < 1$ is a constant.
	
	To protect the true value of states, the agents add noisy signals $w(k)$ into their states $x(k)$ and form a new state vector $x^+(k)$, before sharing with their neighbors, i.e., $x^+(k) = x(k) + w(k)$.
\end{algorithm}

\begin{remark}
	According to \cite{Mo 17 }, the privacy-preserving average consensus algorithm above guarantees the initial state privacy of every benign agent, under mild conditions, and that random noises introduced to the consensus process do not affect the consensus result. 
	%		Moreover, it guarantees that each agent's state will convergence to the average of all agents? initial states in mean-square(in $ L^2 $) .
\end{remark}
Under this privacy-preserving algorithm, since the set of these unknown agents  that are both ``malicious" and ``curious" are $ \{i_1,\cdots,i_p\} $, the state updating rule is as follows:
\begin{equation}
\begin{aligned}
x(k+1) =Ax^+(k)+B u(k) 
=A(x(k)+w(k))+B u(k)   \label{x(k+1)},
\end{aligned}
\end{equation}
where $B=[e_{i_1},e_{i_2},\cdots,e_{i_p}]$ with $e_i$ being the $ i $ th canonical basis vector in $\mathbb{R}^n$, and $u(k)=[u_{i_1}(k), u_{i_2}(k), \cdots, u_{i_p}(k)]^\top $ is the attack input signal at time step $ k $.

\begin{theorem}[\cite{Mo 17 }]\label{privacy}
	The initial state value $ x_j(0) $ of agent $ j $ is kept private from these curious attackers $ \{i_1,\cdots,i_p\} $ if and only if $ \mathcal{N}_j \cup \{j\} \nsubseteq  \mathcal{N}_{i_1} \cup \cdots \cup \mathcal{N}_{i_p} \cup \{i_1,\cdots, i_p\}$.
\end{theorem}

\section{Security  Protection Against Malicious Attackers}
In this section, we address the problem of malicious attackers distributing the average consensus or preventing consensus.
\subsection{Attack Detector}

In order to deal with malicious attacks, we will design an attack detector for each benign agent. 
Without of loss generality, assume that \textit{agent $1$ is benign}, and we focus on designing an attack detector for agent $ 1 $. 
Suppose the neighbors of agent $ 1 $ are $\{j_1,j_2,\cdots,j_{m-1}\}$. 
The values that are available to agent $ 1 $ at $ k$ time step will be denoted by
\begin{equation}
y(k)=C(x(k)+w(k))  	\label{y(k)},
\end{equation}
where $C=[e_1, e_{j_1},e_{j_2},\cdots,e_{j_{m-1}}]^\top $.

We first propose a residual generator as follow, which uses the measurement sequence $\{y(k)\}_{k = 0}^{\infty}$ to generate a residual vector sequence $\{r(k)\}_{k=0}^{\infty}$ that will be a zero vector sequence when there is no noise protecting the agent's privacy and malicious attackers in the system. 
The response of linear consensus system of the form (\ref{x(k+1)}) and (\ref{y(k)}) over $ n+1 $ times steps at each time step $ k $ is given by
\begin{multline}\label{response}
\underbrace{\begin{bmatrix}
	y(k) \\
	y(k+1) \\
	\vdots \\
	y(k+n)
	\end{bmatrix}}_{Y_{[k,k+n]}} =\underbrace{\begin{bmatrix}
	C \\
	C A \\
	\vdots \\
	C A^{n}
	\end{bmatrix}}_{\mathcal{O}_{n}} x(k) \\
+ 
\underbrace{
	\begin{bmatrix}
	C & \mathbf{0} & \cdots & \mathbf{0} \\
	CA & C& \cdots & \mathbf{0} \\
	\vdots & \vdots & \ddots & \vdots \\
	CA^{n} & CA^{n-1} & \cdots & C
	\end{bmatrix}
}_{\mathcal{H}_{n}}
\underbrace{\begin{bmatrix}
	w(k) \\
	w(k+1) \\
	\vdots \\
	w(k+n)
	\end{bmatrix}}_{W_{[k,k+n]}} \\
+ 
\underbrace{
	\begin{bmatrix}
	\mathbf{0} & \mathbf{0} & \cdots & \mathbf{0} \\
	CB & \mathbf{0} & \cdots & \mathbf{0} \\
	\vdots & \vdots & \ddots & \vdots \\
	CA^{n-1}B & CA^{n-2}B & \cdots & \mathbf{0}
	\end{bmatrix}
}_{\mathcal{J}_{n}} 
\underbrace{\begin{bmatrix}
	u(k) \\
	u(k+1) \\
	\vdots \\
	u(k+n)
	\end{bmatrix}}_{U_{[k,k+n]}}.
\end{multline}

Now, we are ready to proceed with the construction of residual generator to be used to design an attack detector. In order to make the residual generator not affected by the initial state value $ x(0) $, we multiply the orthogonal projection matrix $ \mathcal{P}  = I_{m(n+1)} - \mathcal{O}_{n}\mathcal{O}_{n}^+$ on both the left and right sides of the equality \eqref{response}, and equality \eqref{response} can be simplified to the following form:
\begin{equation}\label{response_1}
\mathcal{P}Y_{[k,k+n]} =  \mathcal{P}\mathcal{H}_{n} W_{[k,k+n]}+ \mathcal{P}\mathcal{J}_{n} U_{[k,k+n]}.
\end{equation}
Up to now, based on the above results, we can construct the following residual generator, then use it to design an attack detector.
\begin{definition}[Residual Vector]\label{attack detector}
	Define the residue vector $ r(k) $ at each time step $ k $ as shown below:
	
	\begin{equation}\label{r(k)}
	r(k) \stackrel{\bigtriangleup}{=} \mathcal{P} Y_{[k,k+n]}.
	\end{equation}
	Then a malicious attack detector is obtained, which compares $\|r(k)\|$ with a threshold $ c\rho^k $ decreasing exponentially over time and triggers an alarm if and only if $\|r(k)\|$ is greater than the given threshold $ c\rho^k $ at some time step $ k $,
	% Specifically, this can be represented as
	%	\[
	%	\left\{\begin{array}{cc}
	%	\|r(k)\| >  c\rho^k, & \textrm{Triggering an alarm}\\
	%	\|r(k)\| \leq c\rho^k, & \textrm{Not triggering an alarm}
	%	\end{array}
	%	\right. ,
	%	\]
	where $c > 0 $ , $\varphi < \rho < 1 $ are two fixed constants selected by agent $ 1 $.
\end{definition}
If there is an alarm is triggered at some time step $ k $, we will think that there are malicious attackers in the system.

\begin{remark}
	%		According to the definition of residue vector $ r(k) $, i
	It can be seen that there is a delay of $ n $ time steps in the process of detecting at each time step $ k $. In general, the delay is inevitable, because at any time step $ k $, using the observations up to the current moment cannot get enough information about the attack input at the current moment.
\end{remark}

%\begin{remark}
%	In the design of the attack detector, the reason for designing an exponentially decreasing threshold $ c\rho^k $ for the residual vector $ r(k) $ is that the variance of noisy signal $ w(k) $, which obeys the normal distribution, added to the system at every time step $ k $ is exponentially decreasing.
%\end{remark}
Note that according to the definition of the residual detector, it can be seen that in the absence of noisy signals preserving privacy, that there is no malicious attacker in the system implies that the residual vector $ r(k) \equiv \mathbf{0}_{ m(n+1) \times 1} $ for every $ k \in \mathbb{N} $, but the opposite is not necessarily true. This situation is undesirable because in this case, there exists some attack input $ u $ such that the attack detector does not detect its existence. 	
We will discuss in detail how to avoid this situation in the third subsection detectability.
\subsection{False Alarm Rate}
In this subsection, we will focus on the situation where there is no malicious attacker in the system.
According to the privacy-preserving consensus algorithm, the noisy signal $ w(k) $ will be added into agents' states $ x(k) $ at each time step $ k $, and therefore, even if there is no malicious attacker, an alarm may be triggered at some time step $ k $. 
False alarm rate of the attack detector will be characterized here.

According the linearity of linear consensus system of the form (\ref{x(k+1)}) and (\ref{y(k)}) and the definition of residual vector $ r(k) $, $r(k)$ can be decomposed into the following two parts:
\begin{equation}
r(k) = r^a(k) + r^n(k),
\end{equation}
where $ r^a(k) $ and $ r^n(k) $ are respectively generated by malicious attackers' input and the noisy signals.
By using \eqref{response_1}, we can directly get the following equivalent relationship:
\begin{equation}\label{r(k)_0_a_n}
r^a(k) = \mathcal{P}\mathcal{J}_{n} U_{[k,k+n]}, \ \ \ \
r^n(k) = \mathcal{P}\mathcal{H}_{n} W_{[k,k+n]}.
\end{equation}
Now we define false alarm rate as follows.
\begin{definition}[False Alarm Rate] \label{False Alarm Rate}
	Define false alarm rate $\alpha$ as the probability of triggering false alarm at least once from the initial time step
	to infinity when there is no malicious attacker in the linear consensus system of the form (\ref{x(k+1)}) and (\ref{y(k)}), in other words,
	\begin{equation}
	\alpha \stackrel{\bigtriangleup}{=} \mathbb{P} \left[ \bigcup_{k=0}^{\infty} \left\{ \| r^n(k) \| > c\rho^k
	\right\} \right] \label{rate}.
	\end{equation}
\end{definition}
\begin{remark}
	Under the same conditions, a smaller $ \alpha $ often means better performance of corresponding attack detector.
\end{remark}
Before characterizing false alarm rate $ \alpha $, we first focus on $ r^n(k) $. For the convenience of notation, the matrix $ \mathcal{P}\mathcal{H}_{n} $ is uniformly partitioned according to the columns as $ \begin{bmatrix}
\mathcal{P}_0 & \mathcal{P}_1 & \cdots & \mathcal{P}_{n}
\end{bmatrix}$,
where each $ \mathcal{P}_i $ are of dimension $ m(n+1) \times n $ and $ \mathcal{P}_0 =  \mathcal{P} \mathcal{O}_{n} = \mathbf{0}_{p \times n}$.
Combined with the definition of $ w(k) $, for any time step $ k $, $ r^n(k) $ can then be expressed again as 
\begin{equation}\label{r^n(k)}
r^n(k) =  \sum_{i = 0}^{n-1} \varphi^{k+i} (\mathcal{P}_i - \mathcal{P}_{i+1}) v(k+i) + \varphi^{k+n} \mathcal{P}_{n} v(k+n).
\end{equation}

\begin{theorem}[An Estimation Of False Alarm Rate]  \label{An Estimation of False Alarm Rate}
	For a linear consensus system of the form (\ref{x(k+1)}) and (\ref{y(k)}), false alarm rate $ \alpha  $ of the attack detector above satisfies
	\begin{multline}\label{alpha_upper}
	\alpha \leq \frac{1}{c^2}	\frac{\rho^2}{\rho^2 - \varphi^2}   \left(
	\sum_{i = 0}^{n-1}  \varphi^{2i} {\rm tr}\left[(\mathcal{P}_i - \mathcal{P}_{i+1})^\top (\mathcal{P}_i - \mathcal{P}_{i+1}) \right] \right.\\
	+ \varphi^{2n} {\rm tr}\left[ \mathcal{P}_{n}^\top \mathcal{P}_{n} \right]
	\Bigg)  .
	\end{multline}
\end{theorem}
\begin{proof}
	According to the definition of false alarm rate $ \alpha $, we can express $ \alpha $ as follow:
	\begin{equation}\label{alpha_proof}
	\begin{aligned}
	\alpha &= \mathbb{P} \left[ \bigcup_{k=0}^{\infty} \left\{ \|r^n(k) \| > c\rho^k
	\right\} \right] \leq \sum_{k=0}^{\infty} \mathbb{P} \left[ \|r^n(k) \| > c\rho^k \right] \\
	&\leq \sum_{k=0}^{\infty}  \frac{\mathbb{E}\left[r^n(k)^\top r^n(k) \right]}{c^2 \rho^{2k}} 
	= \frac{\eta}{c^2}  \sum_{k=0}^{\infty}  \frac{\varphi^{2k}}{\rho^{2k}} 
	= \frac{\eta}{c^2}\frac{\rho^2}{\rho^2 - \varphi^2},
	\end{aligned}
	\end{equation}
	where $ \eta  $ represents the item in braces on the right side of \eqref{alpha_upper}. The first inequality holds because of the countable additivity of probability measures and the second inequality holds because of Chebyshev's inequality.		
\end{proof}

\subsection{Detectability}
In this subsection, we will focus on the detectability of the attack detector.
%First, we give the definition of  undetectable input as follow.

\begin{definition}[Undetectable Input] 
	For a linear consensus system of the form (\ref{x(k+1)}) and (\ref{y(k)}), the attack input $ u $ introduced by 
	these unknown malicious agents $ \{i_1,\cdots,i_p\} $ is undetectable if 
	\[ \exists x_1,x_2 \in \mathbb{R}^n, {\rm s.t. } \forall k \in \mathbb{N}, y^{0+a}(x_1,u,k) = y^{0+a}(x_2,0,k),
	\]
	where $  y^{0+a}(x_1,u,k) $ is the part of $ y(x_1,u,k) $ generated by the initial state $ x_1 $ and the attack input $ u $ at time step $ k $.
	%    	For a linear consensus system of the form (\ref{x(k+1)}) and (\ref{y(k)}), the attack input $ u $ introduced by 
	%    	these unknown malicious agents $ \{i_1,\cdots,i_p\} $ is undetectable by the attack detector if 
	%    	\begin{multline}
	%    	u \text{ is a nonzero input}, \exists \text{ two initial state values } x_1,x_2 \in \mathbb{R}^n : \\
	%    	\forall k \in \mathbb{N}, r(x_1,u,k)  \text{ and } r(x_2,\mathbf{0},k) \text{ obey the same distribution }\\
	%    	\text{law}, i.e.,  r^a(x_1,u,k) =   r^a(x_2,\mathbf{0},k), \label{undetectable 2}
	%    	\end{multline}
	%    	where $ u $ is a nonzero input means there exists some time step $ k' $ such that $ u(k') \neq \mathbf{0}_{p \times 1} $, and
	%    	$r(x_1,u,k)$ represents the residual vector generated by agent $1$ at the time step $ k $ when the initial state is $x_1$, the attack input sequence is $\{u(k)\}_{k=0}^{\infty}$ and the noisy signal sequence is $\{w(k)\}_{k=0}^{\infty}$, and $r^a(x_1,u,k)$ is the partition of $r(x_1,u,k)$ generated by the attack input. The meaning of $ r^a(x_2,\mathbf{0},k) $ is similar to $ r^a(x_1,u,k) $.
\end{definition}
%	\begin{remark}
%		Note that the definition of undetectable input given here is formally different from that given in \cite{Bullo F}, but they represent the same meaning.
%	\end{remark}

\begin{definition}[Undetectable Input By The Attack Detector]
	For a linear consensus system of the form (\ref{x(k+1)}) and (\ref{y(k)}), the attack input $ u $ introduced by 
	these unknown malicious agents $ \{i_1,\cdots,i_p\} $ is undetectable by the attack detector
	if 
	\[ \exists x_1,x_2 \in \mathbb{R}^n, {\rm s.t. } \forall k \in \mathbb{N}, r^{a}(x_1,u,k) = r^{a}(x_2,0,k).
	\]
\end{definition}

Before giving the necessary and sufficient conditions for there exists no undetectable input by the attack detector in the linear consensus system of the form \eqref{x(k+1)} and \eqref{y(k)}, we need the following three lemmas.

\begin{lemma}\label{Detect_before}
	For a linear consensus system of the form (\ref{x(k+1)}) and (\ref{y(k)}), if the first $ n $ columns of  $  \begin{bmatrix}
	\mathcal{O}_{n-1} & \mathcal{J}_{n-1}
	\end{bmatrix} $ are independent of each other and the last $ np $ columns 
%	of $  \begin{bmatrix}
%	\mathcal{O}_{n-1} & \mathcal{J}_{n-1}
%	\end{bmatrix} $,
	 i.e., $
	{\rm rank} \begin{bmatrix}
	\mathcal{O}_{n-1} & \mathcal{J}_{n-1}
	\end{bmatrix} - {\rm rank} [ \mathcal{J}_{n-1}]
	= n$,
	then the first $ p $ columns of $ \mathcal{J}_{n-1} $ are independent of each other and the last $ np $ columns 
%	of $  
%	\mathcal{J}_{n-1} $, 
	i.e., $
	{\rm rank}  [ \mathcal{J}_n ]
	- {\rm rank} [ \mathcal{J}_{n-1}] = p$,
	where the specific relationship between $ \mathcal{J}_{n} $ and $  \mathcal{J}_{n-1} $ can be expressed as 
	$ \mathcal{J}_{n} = \begin{bmatrix}
	\mathbf{0}_{m \times p} & \mathbf{0}_{m \times np} \\
	\mathcal{O}_{n-1}B &  \mathcal{J}_{n-1}
	\end{bmatrix} $.
\end{lemma}
\begin{proof}
	Just notice that $ B = [e_{i_1},e_{i_2},\cdots,e_{i_p}] $ is column full rank.
\end{proof}

\begin{lemma}[\cite{Sundaram2010}]\label{strongly observable}
	A linear consensus system of the form (\ref{x(k+1)}) and (\ref{y(k)}), is said to be strongly observable, if $ y^0(k) + y^a(k) \equiv \mathbf{0}_{m \times 1} $ for all $ k \in \mathbb{N} $ implies $ x(0) = \mathbf{0}_{n \times 1} $(regardless of the values of the input $ u $), where 
	$ y^0(k) $ and $ y^a(k) $ are respectively the part of $ y(k)  $ generated by initial state $ x(0) $ and that generated by malicious input $ u $. The following statements are equivalent.
	\begin{enumerate}
		\item 	$
		{\rm rank} \begin{bmatrix}
		\mathcal{O}_{n-1} & \mathcal{J}_{n-1}
		\end{bmatrix} - {\rm rank} [ \mathcal{J}_{n-1}]
		= n$;
		\item the system is strongly observable.
	\end{enumerate}
\end{lemma}

\begin{lemma}\label{observable}
	For a linear consensus system of the form (\ref{x(k+1)}) and (\ref{y(k)}), it is observable, i.e, $ {\rm rank}[\mathcal{O}_{n-1}] = n  $, almost surely.
\end{lemma}
\begin{proof}
	This is a direct corollary of Theorem 2 in \cite{Sundaram2008_3}.
	%		Since the network $ G $ considered in this paper is connected, 
	%		for any agent $ j $, there exists a path from agent $ j $ to agent $ 1 $. According to Theorem 2 of \cite{Sundaram2008_3}, for almost any\footnote{The phrase ``almost any'' indicates that the set of parameters for which Lemma does not hold has Lebesgue measure zero.} choice of weight matrix $ A $, agent $ 1 $ can obtain the initial value $ x_j(0) $ after performing the update rule \eqref{eq1} for $ n-1 $ time steps when there is no malicious attackers and noisy signals preserving initial state privacy in the system. Note that $ Y_{[0,n-1]} = \mathcal{O}_{n-1}x(0) $, therefore, we can get that $ {\rm rank}[\mathcal{O}_{n-1}] = n $ holds almost surely.
\end{proof}
%	In this paper, we assume $ A $ satisfies that condition in \cite{Sundaram2008_3} such that for agent $ 1 $, $ {\rm rank}[\mathcal{O}_{n-1}] = n  $

\begin{theorem} \label{Detect}
	For a linear consensus system of the form (\ref{x(k+1)}) and (\ref{y(k)}), the following statements are equivalent almost surely:
	\begin{enumerate}
		\item there is no undetectable input;
		\item there is no undetectable input by the attack detector;
		\item 	$
		{\rm rank} \begin{bmatrix}
		\mathcal{O}_{n-1} & \mathcal{J}_{n-1}
		\end{bmatrix} - {\rm rank} [ \mathcal{J}_{n-1}]
		= n.$
	\end{enumerate}
\end{theorem}

\begin{proof}
	($ 2 \Rightarrow 3 $): Suppose that $ 	{\rm rank} \begin{bmatrix}
	\mathcal{O}_{n-1} & \mathcal{J}_{n-1}
	\end{bmatrix} - {\rm rank} [ \mathcal{J}_{n-1}]
	<n $, according to Lemma \ref{strongly observable}, the system is not strongly observable. Therefore, there exists a nonzero $ x(0) \in \mathbb{R}^n $ and an attack input $ u $ such that $ y^0(k) + y^a(k) \equiv \mathbf{0}_{m \times 1} $ for all $ k \in \mathbb{N} $. According to the definition of $ r(k) $, there must be that $ r^a(k)  \equiv \mathbf{0} $ for all $ k \in \mathbb{N} $. Since there is no undetectable input by the attack detector, it follows that 
	\begin{multline}\label{key}
	r^a(k)  \equiv \mathbf{0}, \forall k \in \mathbb{N}
	\Rightarrow u(k) \equiv \mathbf{0}, \forall k \in \mathbb{N} \\
	\Rightarrow y^a(k) \equiv \mathbf{0}, \forall k \in \mathbb{N}
	\Rightarrow y^0(k) \equiv \mathbf{0}, \forall k \in \mathbb{N},
	\end{multline}
	where the last step holds because $ y^0(k) + y^a(k) \equiv \mathbf{0}_{m \times 1} $ for all $ k \in \mathbb{N} $.
	Combined with Lemma \ref{observable}, we can get the following relationship almost surely:
	\begin{equation}\label{key}
	y^0(k) \equiv \mathbf{0}, \forall k \in \mathbb{N} \stackrel{{\rm rank}[\mathcal{O}_{n-1}] = n}{\Longrightarrow}
	x(0) = \mathbf{0},
	\end{equation}
	but it contracts the fact that $ x(0) $ is a nonzero vector.
	
	($ 1 \Rightarrow 3 $): The proof is similar to that of ($ 2 \Rightarrow 3 $).
	
	($ 3 \Rightarrow 2 $):
	First, we assert that there is a $ p \times m(n+1) $ matrix $ \mathcal{Q}_B $ \footnote{The subscript $B$ here means that the matrix $ \mathcal{Q}_B $ is related to $B$. 
%		Note that $\mathcal{Q}_B$ may not be unique. Here we select the one with the smallest  2-norm among all as $\mathcal{Q}_B$ whenever them exist.
	} that satisfies $ 	\mathcal{Q}_B\mathcal{P}\mathcal{J}_{n}  = 
	\left[
	\begin{array}{c|c}
	I_p & \mathbf{0}_{p \times np}
	\end{array}
	\right]$.
	Otherwise, it implies that the statement that the first $ p $ columns of $ \mathcal{P}\mathcal{J}_{n} $ are linearly independent of each other and the last $ np $ columns is false. Therefore, there exists a nonzero input sequence $ \{u(i)\}_{i=0}^{n} $ with $ u(0) \neq \mathbf{0}_{p \times 1} $ such that $\mathcal{P}\mathcal{J}_{n}U_{[0,n]} = \mathbf{0}_{m(n+1) \times 1}  $. 
	\begin{itemize}
		\item If $ \mathcal{J}_{n}U_{[0,n]} \neq \mathbf{0}_{m(n+1) \times 1}  $, according to the definition of $ \mathcal{P} $, there must exist a nonzero initial state $ x(0) \in \mathbb{R}^n $ such that $\mathcal{O}_{n} x(0) + \mathcal{J}_{n} U_{[0,n]} = \mathbf{0}_{m(n+1) \times 1}   $, but this contracts the result $ 	{\rm rank} \begin{bmatrix}
		\mathcal{O}_{n} & \mathcal{J}_{n}
		\end{bmatrix} - {\rm rank} [ \mathcal{J}_{n}]
		= n $, which is a direct corollary of the condition ${\rm rank} \begin{bmatrix}
		\mathcal{O}_{n-1} & \mathcal{J}_{n-1}
		\end{bmatrix} - {\rm rank} [ \mathcal{J}_{n-1}]
		= n$. 
		\item If $ \mathcal{J}_{n}U_{[0,n]} = \mathbf{0}_{m(n+1) \times 1}  $, however, it contracts the results of Lemma \ref{Detect_before}, for $ u(0) \neq \mathbf{0}_{p \times 1} $.
	\end{itemize}
	The above result shows that the initial assertion is correct.
	Since $ r^a(k) = \mathcal{P}\mathcal{J}_{n} U_{[k,k+n]} $, we can get that 
	\begin{equation}\label{Q_B}
	\mathcal{Q}_Br^a(k) = u(k) .
	\end{equation}
	Therefore, $	r^a(k) \equiv \mathbf{0}, \forall k \in \mathbb{N} \Rightarrow u(k) \equiv \mathbf{0}, \forall k \in \mathbb{N}$, i.e., there is no undetectable input by the attack detector.
	
		($ 3 \Rightarrow 1 $): Suppose there is a nonzero attack input $ u $ and an initial state $ x(0)  $
		 such that $ y^{0+a}(k) \equiv \mathbf{0} $ for all $ k \in \mathbb{N} $. Since $
		 {\rm rank} \begin{bmatrix}
		 \mathcal{O}_{n-1} & \mathcal{J}_{n-1}
		 \end{bmatrix} - {\rm rank} [ \mathcal{J}_{n-1}]
		 = n$, there must be that $ x(0) = \mathbf{0} $. Further, it follows that
		 \begin{multline}\label{key}
		 x(0) = \mathbf{0} \Rightarrow y^0(k) \equiv \mathbf{0}, \forall k \in \mathbb{N} \Rightarrow  y^a(k) \equiv \mathbf{0}, \forall k \in \mathbb{N} \\
		  \Rightarrow  r^a(k)\equiv \mathbf{0}, \forall k \in \mathbb{N} \Rightarrow  u(k)\equiv \mathbf{0}, \forall k \in \mathbb{N},
		 \end{multline}
		 where \eqref{Q_B} in ($ 3 \Rightarrow 2 $) is used in the last step. This contracts that the fact that $ u $ is nonzero. Therefore, there is no undetectable input.
\end{proof}

\subsection{Asymptotic Consensus and Error}
In order to protect the security and privacy of the system, we have designed a privacy-preserving average consensus algorithm equipped with an attack detector with a time-varying exponentially decreasing threshold for every benign agent.
At this point, there are naturally three problems: 
\begin{itemize}
	\item When there exists no undetectable inputs and no alarm is triggered from beginning to end, will the system eventually achieve consensus?
	\item   If the system can achieve a final consensus, what is the rate of convergence?
	\item How much error of the final consensus value will be?
\end{itemize}
We will answer these three problems in turn in this subsection.

\begin{lemma}\label{decomposition}(\cite{Mo 17 })
	Define a matrix $\mathcal{A} \stackrel{\bigtriangleup}{=} A - \frac{\mathbf{1}\mathbf{1}^\top }{n}$. 
	For $\forall k \in \mathbb{N}$, $ \mathcal{A}^k = A^k - \frac{\mathbf{1}\mathbf{1}^\top }{n} $.
\end{lemma}
\begin{proof}
	Just notice that by assumption (A1) and (A2), we can get that
	$ \frac{\mathbf{1}\mathbf{1}^\top }{n}A = \frac{\mathbf{1}\mathbf{1}^\top }{n} = A\frac{\mathbf{1}\mathbf{1}^\top }{n}$.
	%		the following equalities hold:
	%		\begin{equation}
	%		\frac{\mathbf{1}\mathbf{1}^\top }{n}A = \frac{\mathbf{1}\mathbf{1}^\top }{n} = A\frac{\mathbf{1}\mathbf{1}^\top }{n}.
	%		\end{equation}
	%		As a result, $\mathcal{A}^k = A^k - \frac{\mathbf{1}\mathbf{1}^\top }{n}$. 
\end{proof}

\begin{lemma}\label{u<infty}
	For a linear consensus system of the form (\ref{x(k+1)}) and (\ref{y(k)}), the following inequality holds almost surely
	\begin{equation}\label{u(k)_conver}
	\sum_{k=0}^{\infty} \frac{\| u(k) \|}{(\rho + \varepsilon)^k}  < \infty ,
	\end{equation}
	with $ \varepsilon $ is a fixed constant satisfying $ 0 <  \varepsilon < 1 - \rho $, if 
	\begin{enumerate}
		\item no alarm is triggered;
		\item $ 	
		{\rm rank} \begin{bmatrix}
		\mathcal{O}_{n-1} & \mathcal{J}_{n-1}
		\end{bmatrix} - {\rm rank} [ \mathcal{J}_{n-1}]
		= n
		$.
	\end{enumerate}
\end{lemma}
\begin{proof}
	Since no alarm is triggered, by the definition of  ``not triggering an alarm'', it follows that $ \|r^n(k) + r^a(k) \| \leq c \rho^k$. By using triangle inequality, we can get that $ \|r^a(k)\| \leq c\rho^k + \|r^n(k) \| $. Now we first focus on $ r^n(k) $. According to \eqref{r^n(k)}, there exists a fixed constant $  d > 0$ such that 
	$\|r^n(k)\| \leq d \varphi^k \sum_{i=0}^{n} \| v(k+i)\|$,
	where $ d$ can be selected as $ \sum_{i = 0}^{n-1} \varphi^i \| \mathcal{P}_i - \mathcal{P}_{i+1}\| + \varphi^n \| \mathcal{P}_n \| $.
	Therefore, it is not difficult to get that 
	\begin{equation}\label{r^n(k)_upper}
	\sum_{k=0}^{\infty} \frac{\|r^n(k)\|}{(\rho+ \varepsilon)^k} \leq (n+1)d \sum_{k=0}^{\infty} \left(\frac{\varphi}{\rho+\varepsilon}\right)^k \|v(k)\|.
	\end{equation}
	By Chebyshev's inequality, for any positive integer $ k $, we have
	\begin{equation}\label{key}
	\mathbb{P} \left[  \|v(k)\| \geq  k  \right]
	\leq  \frac{\mathbb{E}[v(k)^\top  v(k) ]}{k^2} = \frac{n}{k^2} \\
	\end{equation}
	and consequently it implies that 
	\begin{equation}\label{key}
	\sum_{k=1}^{\infty}  \mathbb{P} \left[  \|v(k)\| \geq  k  \right] \leq 	\sum_{k=1}^{\infty} \frac{n}{k^2} = \frac{n\pi^2}{6} < \infty.
	\end{equation}
	By Borel-Cantelli Lemma, it follows that 
	\begin{equation}\label{borel}
	\mathbb{P} \left[  \limsup_{k \rightarrow \infty} \left\{ \|v(k)\| \geq  k \right\}  \right] = 0.
	\end{equation}
	Since a point belongs to $  \limsup_{k} \left\{ \|v(k)\| \geq  k \right\} $ if and only if it belongs to infinitely many terms of the sequence $ \left\{ \|v(k)\| \geq  k \right\}_{k = 1}^{\infty}  $, this sequence of $ \left\{ \|v(k)\| \geq  k \right\}$ occurs at most a finite number of times almost surely. Therefore, there exists a positive integer $ k_1 $ such that $ \forall k > k_1 $, $ 	\| v(k) \| < k $ holds almost surely\footnote{Hereafter, ``almost surely" will be abbreviated as ``a.s." sometimes.}.
	It follows that 
	\begin{equation}\label{||v(k)||}
	\begin{aligned}[alignment]
	&\sum_{k= 0}^{\infty} \left( \frac{\varphi}{\rho+ \varepsilon} \right)^k \|v(k)\|  =  \left(
	\sum_{k= 0}^{k_1} + 	\sum_{k= k_1+1}^{\infty} 
	\right)	 \left( \frac{\varphi}{\rho+ \varepsilon} \right)^k \|v(k)\| \\
	&\stackrel{{\rm a.s.}}{<} 	\sum_{k= 0}^{k_1} \left( \frac{\varphi}{\rho+ \varepsilon} \right)^k \|v(k)\| + \sum_{k= k_1+1}^{\infty} k \left(  \frac{\varphi}{\rho+ \varepsilon} \right)^k < \infty,
	\end{aligned}
	\end{equation}
	where the last inequality holds because $0< \varphi < \rho $.
	Combined with \eqref{r^n(k)_upper}, we have
	$ \sum_{k=0}^{\infty} \frac{\| r^n(k) \|}{(\rho + \varepsilon)^k} \stackrel{{\rm a.s.}}{<}  \infty  $. 
	
	Since $ \|r^a(k)\| \leq c\rho^k + \|r^n(k) \| $ holds for all $ k \in \mathbb{N} $ when no alarm is triggered, then we have
	\begin{equation}\label{key}
	\sum_{k=0}^{\infty} \frac{\| r^a(k) \|}{(\rho + \varepsilon)^k}   
	\leq 	\sum_{k=0}^{\infty} \left( c \left(\frac{\rho}{\rho+\varepsilon}\right)^k + 
	\frac{\| r^n(k) \|}{(\rho + \varepsilon)^k}
	\right)   \stackrel{{\rm a.s.}}{<} 
	\infty.
	\end{equation}
	Note that for any $ k $, we have $ u(k) = \mathcal{Q}_Br^a(k) $, where $ \mathcal{Q}_B $ is defined in Theorem \ref{Detect}, this implies that $ \|u(k)\| \leq \| \mathcal{Q}_B\| \|r^a(k)\| $.
	Therefore, \eqref{u(k)_conver} holds almost surely.
\end{proof}
%In order to characterize the convergence rate of the system, we need to give the definition of convergence rate here.
Now, we define convergence rate here.

\begin{definition}[Convergence Rate]
	Define the convergence rate $ \varrho $ of consensus algorithm as
	\begin{equation}\label{key}
	\varrho  \stackrel{\bigtriangleup}{=} \limsup_{k \rightarrow \infty} \left\| x(k) - \overline{x}(k) \right\|^{\frac{1}{k}},
	\end{equation}
	whenever the limit on the right-hand side exists, where $ \overline{x}(k)  = \frac{\mathbf{1}\mathbf{1}^\top }{n}x(k)$ denotes the average state vector at time step $ k $.
\end{definition}

\begin{theorem} \label{converage}
	For a linear consensus system of the form (\ref{x(k+1)}) and (\ref{y(k)}), an asymptotic consensus will be reached almost surely, i.e. $\lim_{k\rightarrow \infty} x(k) - \overline{x}(k) \stackrel{{\rm a.s.}}{=} \mathbf{0}_{n \times 1}$, and the convergence rate $ \varrho $ satisfies $
	\varrho \leq  \max\{\rho,|\lambda_2|,|\lambda_n| \}$, 
	if 
	\begin{enumerate}
		\item no alarm is triggered;
		\item $ 	
		{\rm rank} \begin{bmatrix}
		\mathcal{O}_{n-1} & \mathcal{J}_{n-1}
		\end{bmatrix} - {\rm rank} [ \mathcal{J}_{n-1}]
		= n
		$.
	\end{enumerate}
\end{theorem}
\begin{proof}
	Using the result of  Lemma \ref{decomposition}, we can get that 
	\begin{equation}\label{x-x_ave}
	x(k) - \overline{x}(k) = \mathcal{A}^kx(0) + \sum_{i=0}^{k-1} \mathcal{A}^{k-i}w(i) + \sum_{i=0}^{k-1} \mathcal{A}^{k-1-i} B u(i).
	\end{equation}
	Now, we analyze the three terms on the right-side of the equality above in turn.
	
	(1) ``$ \mathcal{A}^kx(0) $'' : 
	For any initial state value $ x(0) $, we have
	$ \| \mathcal{A}^kx(0) \| \leq \max\{|\lambda_2|, |\lambda_n| \}^k \|x(0) \|  \rightarrow 0$ as $ k \rightarrow \infty$.
	
	(2) ``$ \sum_{i=0}^{k-1} \mathcal{A}^{k-i}w(i)  $'': According to Algorithm \ref{algorithm}, we can re-express $ \sum_{i=0}^{k-1} \mathcal{A}^{k-i}w(i)  $ as follow
	\begin{multline}\label{key}
	\sum_{i=0}^{k-1} \mathcal{A}^{k-i}w(i)  = \mathcal{A}\varphi^{k-1}v(k-1) \\
	+ \sum_{i=0}^{k-2}\varphi^i \mathcal{A}^{k-1-i}(\mathcal{A}-I)v(i).
	\end{multline}
	Similar to before, we also have $ \| \mathcal{A}\varphi^{k-1}v(k-1) \| \leq \max \{|\lambda_2|,|\lambda_n|\} \varphi^{k-1} 
	\|v(k-1)\|
	\stackrel{{\rm a.s.}}{\rightarrow} 0 $ as $ k \rightarrow \infty$ similarly. 
	Based on the result given by \eqref{borel} and the triangle inequality, there exists a constant $ d_1 > 0$ such that for any time step $ k $,
	\begin{equation}\label{36}
	\left\| \sum_{i=0}^{k-2}\varphi^i \mathcal{A}^{k-1-i}(\mathcal{A}-I)v(i) \right\| 
	\stackrel{\text{a.s.}}{\leq } d_1 k^2 \max \{\varphi,|\lambda_2|,|\lambda_n|\}^k.
	\end{equation}	
	Since $ 0 < \max\{\varphi, |\lambda_2|,|\lambda_n|\} < 1 $, if let $ k $ approach infinity, it follows that
	$ \lim_{k \rightarrow \infty}	\left\|	 \sum_{i=0}^{k-2}\varphi^i \mathcal{A}^{k-1-i}(\mathcal{A}-I)v(i) \right\| \stackrel{{\rm a.s.}}{=}  \mathbf{0}_{n \times 1} $.
	Therefore, we can get that
	$\sum_{i=0}^{k-2}\varphi^i \mathcal{A}^{k-1-i}(\mathcal{A}-I)v(i)$ convergences to $\mathbf{0}_{n \times 1}$ a.s..
	Combining with the previous results, further, we can get that $  	\sum_{i=0}^{k-1} \mathcal{A}^{k-i}w(i)$ convergences to $\mathbf{0}_{n \times 1} $ a.s. .
	
	(3) ``$ \sum_{i=0}^{k-1} \mathcal{A}^{k-1-i} B u(i) $'': According to the definition of $ B $ and the triangle inequality, it is not difficult to get that
	\begin{equation}\label{key}
	\left\| \sum_{i=0}^{k-1} \mathcal{A}^{k-1-i} B u(i) \right\| \stackrel{{\rm a.s.}}{\leq}  \sum_{i=0}^{k-1} \max\{|\lambda_2|,|\lambda_n|\}^{k-1-i} \|u(i)\|.
	\end{equation}
	According to the result of Lemma \ref{u<infty}, for any $ 0 < \varepsilon < 1- \rho $, there exists a constant $ d_{\varepsilon} > 0 $ related to $ \varepsilon $ such that $ \|u(k)\| \leq d_{\varepsilon} (\rho+\varepsilon)^k $ holds a.s. for all time step $ k $. Combined with these inequalities above, we can get that 
	\begin{multline}\label{u_sum}
	\left\| \sum_{i=0}^{k-1} \mathcal{A}^{k-1-i} B u(i) \right\| \stackrel{{\rm a.s.}}{\leq}  d_{\varepsilon} \sum_{i=0}^{k-1} \max\{\rho+\varepsilon,|\lambda_2|,|\lambda_n|\}^{k-1} \\
	= d_{\varepsilon}  k \max\{\rho+\varepsilon,|\lambda_2|,|\lambda_n|\}^{k-1} .
	\end{multline}
	If let $ k $ approach infinity, we can directly get that $ 	\sum_{i=0}^{k-1} \mathcal{A}^{k-i}Bu(i) $ convergences to $ \mathbf{0}_{n \times 1}  $ a.s. .
	
	Combining all results above, for any initial state $x(0)$, $ x(k) - \overline{x}(k) $ converges to $\mathbf{0}_{n \times 1}$ almost surely.

	Now we analyze the convergence rate whenever an asymptotic consensus can be reached. For convenience, we analyze the convergence rate of the three terms on the right-side of \eqref{x-x_ave} in turn and note that the final total convergence rate is the largest of these three.
	
	(1) ``$ \mathcal{A}^kx(0) $'' : Obviously, the convergence rate of the first term is $ \max \{|\lambda_2|,|\lambda_n|\} $.
	
	(2) ``$ \sum_{i=0}^{k-1} \mathcal{A}^{k-i}w(i)  $'': Similar to the above, the convergence rate of $  \mathcal{A}\varphi^{k-1}v(k-1) $ is $ \varphi  $.
	According to \eqref{36}, we can get that the convergence rate of the item on left-hand side of inequality in \eqref{36}  is no more than $   \max \{\varphi,|\lambda_2|,|\lambda_n|\} $ whenever the inequality holds.
	
	(3) ``$ \sum_{i=0}^{k-1} \mathcal{A}^{k-1-i} B u(i) $'': According to \eqref{u_sum}, we can directly get that the convergence rate of the last term in right hand side of \eqref{x-x_ave} is no more than $  \max\{\rho+\varepsilon,|\lambda_2|,|\lambda_n|\}  $ for any $ 0 < \varepsilon < 1-\rho $ whenever an asymptotic consensus can be reached. Indeed, this implies that its convergence rate is no more than $ \max\{\rho,|\lambda_2|,|\lambda_n|\} $.
	
	Since $ 0 < \varphi < \rho $, combining with all results above, we can get that when an asymptotic consensus is reached, the convergence rate $ \varrho $ satisfies that $ 	\varrho \leq  \max\{\rho,|\lambda_2|,|\lambda_n| \}$.
\end{proof}

%	\subsection{Estimating The Error}
Now we come to answer the last of these three problems raised at the beginning of this subsection, namely how to estimate the error of the final consensus value. 
%Before officially starting, we first give the specific definition of error of the final consensus value.
\begin{definition}(The Error Of The Final Consensus Value)
	The error  of the final consensus value of the system of the form \eqref{x(k+1)} and \eqref{y(k)} $e$ is defined as follow
	\begin{equation}\label{key}
	e \stackrel{\bigtriangleup}{=} \frac{1}{n}\mathbf{1}_{1 \times p}^\top  \sum_{i = 0}^{\infty}u(i).
	\end{equation}
	%		\begin{equation}
	%		e \stackrel{\bigtriangleup}{=} \frac{1}{n}\mathbf{1}_{1 \times n}^\top  \sum_{i = 0}^{\infty}Bu(i).
	%		\end{equation}
	%		Indeed, the error $ e $ can be re-express as $ e =  \frac{1}{n}\mathbf{1}_{1 \times p}^\top  \sum_{i = 0}^{\infty}u(i)$.
\end{definition}
\begin{remark}
	%		It is important to note that we have shown that the noise protecting agent's privacy has no effect on the final consensus value when the system reaches asymptotic consensus in Theorem \ref{converage}. Therefore, the definition of error of the final consensus value above is well-defined.
	The definition of error above is well-defined, for the noisy signals have no effect on the final consensus value, which has been proved in Theorem \ref{converage}. 
%	Indeed, the error $ e $ can be re-express as $ e =  \frac{1}{n}\mathbf{1}_{1 \times n}^\top  \sum_{i = 0}^{\infty}Bu(i)$.
\end{remark}
It is worth noting that the estimation of error needs to be carried out under the premise that no alarm is triggered from beginning to end in the system. According to the definition of ``not triggering an alarm'', we have $ \|r(k)  \| \leq c \rho^k$ for all time step $ k $. By using \eqref{r(k)_0_a_n}, \eqref{Q_B}, Algorithm \ref{algorithm} and summing the time step $ k $ from $0$ to infinity, we have
\begin{equation}\label{key}
\sum_{k=0}^{\infty} u(k) = \mathcal{Q}_B\footnote{Note that $\mathcal{Q}_B$ may not be unique. Here we select the one such that $ \|\mathbf{1}_{1 \times p}\mathcal{Q}_B\| $ achieves the minimum.} \left(\sum_{k=0}^{\infty} r(k) +   \sum_{i=1}^{n}\varphi^{i-1}\mathcal{P}_i v(i-1)\right).
\end{equation}
Substituting the above equality into the definition of error $ e $, we can get
\begin{equation}\label{e}
e = \frac{\mathbf{1}_{1 \times p} \mathcal{Q}_B }{n}
\left(\sum_{k=0}^{\infty} r(k) +   \sum_{i=1}^{n}\varphi^{i-1}\mathcal{P}_i v(i-1)\right).
\end{equation}
	For convenience and simplicity of notation, let $ s_B $ and $ T_B $ denote $ \frac{\mathbf{1}_{1 \times p} \mathcal{Q}_B }{n}  \sum_{k=0}^{\infty} r(k) $ and $ \sum_{i=1}^{n}\frac{\varphi^{i-1}}{n}\mathbf{1}_{1 \times p} \mathcal{Q}_B  \mathcal{P}_i v(i-1) $ respectively. Since $v_i(k) (i=1,2,\cdots,n; k=0,1,\cdots)$ are standard normal distributed random variables, which are independent across $i$ and $k$, $ T_B $ is also a normal distributed random variable with $ \mathbb{E}[T_B] = 0 $ and $ {\rm Var}[T_B] = \sum_{i = 1}^n \frac{\varphi^{2(i-1)}}{n^2} \left\| \mathbf{1}_{1 \times p} \mathcal{Q}_B\mathcal{P}_i \right\|^2$. For any $ 0  < \beta < 1 $, let the point $ z_{B,\beta/2} $ satisfies $ \mathbb{P}\left[T_B > z_{B,\beta/2} \right]  = \beta/2$. Therefore, we have $ \mathbb{P}\left[|T_B| \leq z_{B,\beta/2} \right]  = 1-\beta $. Since $ e = s_B + T_B $, it follows that $ \mathbb{P}_e \left[|e -s_B| \leq z_{B,\beta/2} \right]  = 1-\beta $ for any $ e \in \mathbb{R} $.

	Now let $ 	\mu_B =  \frac{c}{n(1-\rho)}  \left\| \mathbf{1}_{1 \times p} \mathcal{Q}_B \right\|$.
%	\begin{equation}\label{key}
%	\mu = \frac{c}{n(1-\rho)}\left\|\mathbf{1}_{1 \times p^{*}} \mathcal{Q}_{B^{*}}\right\|,
%%	\sigma^2 =  \frac{\left\|\mathbf{1}_{1 \times p^{*}} \mathcal{Q}_{B^{*}}\right\|^{2}}{n^{2}} \sum_{i=1}^{n} \varphi^{2(i-1)}\left\|\mathcal{P}_{i}\right\|^{2},
%	\end{equation}
%	where $p^{*}$ is the number of columns of $B^{*}$ and $ B^* $ satisfies
%	\begin{equation}\label{key}
%	B^* = \arg \max_{B} \min_{\mathcal{Q}_B} \left\| \mathbf{1}_{1 \times p} \mathcal{Q}_B \right\|.
%	\end{equation}
%	Among them, 
%	where $\max_B$ traverses all $B$s that meet the detectable condition $\operatorname{rank}\left[\mathcal{O}_{n-1} \quad \mathcal{J}_{n-1}\right]-\operatorname{rank}\left[\mathcal{J}_{n-1}\right]=n$.
%	and $\min_{\mathcal{Q}_B}$ traverses all $\mathcal{Q}_B$s satisfying $ 	\mathcal{Q}_B\mathcal{P}\mathcal{J}_{n}  = 
%	\left[
%	I_p \ | \  \mathbf{0}_{p \times np}
%	\right]$. 
	Under the detectable condition $\operatorname{rank}\left[\mathcal{O}_{n-1} \quad \mathcal{J}_{n-1}\right]-\operatorname{rank}\left[\mathcal{J}_{n-1}\right]=n$, ``no alarm is triggered'' implies that $ |s_B| \leq \mu_B $ holds because of  Cauchy-Schwarz inequality. Therefore, we can get that $ \mathbb{P}_e \left[  -\mu_B-z_{B,\beta/2} \leq e \leq \mu_B + z_{B,\beta/2} \right]  \geq 1-\beta $ for any $ e \in \mathbb{R} $, i.e., $ [-\mu_B-z_{B,\beta/2}, \mu_B+z_{B,\beta/2}] $ is a confidence interval for $ e $ with confidence coefficient of not less than $ 1- \beta $.

Based on the results above and Theorem \ref{False Alarm Rate} , we can get the following theorem.
\begin{theorem}\label{error}
	For a linear consensus system of the form (\ref{x(k+1)}) and (\ref{y(k)}), when an asymptotic consensus is reached, for any $ 0< \beta < 1 $, $
	\bigcup_{B}\footnote{$\bigcup_{B}$ traverses all $B$s that meet the detectable condition $ {\rm rank} \begin{bmatrix}
		\mathcal{O}_{n-1} & \mathcal{J}_{n-1}
		\end{bmatrix} - {\rm rank} [ \mathcal{J}_{n-1}]
		= n $. } [-\mu_B-z_{B,\beta/2}, \mu_B+z_{B,\beta/2}] $
	is a confidence interval for $ e $ with confidence coefficient of not less than $ 1- \beta $,
%	there is a probability of not less than $ \beta $ such that $ |e| \leq e_{\beta} $, where $ e_{\beta}  = \max_{B}\footnote{$\max_B$ traverses all $B$s that meet the detectable condition.} e_B$ and $ \Phi  $ is the cumulative distribution function of the standard normal distribution,
%%	
%	$ \forall \delta > 0 $,
%	\begin{equation}\label{key}
%	 \mathbb{P}\left[ \left|e -  \mu  \right| \leq \delta  \right] \geq 2 \Phi\left(\frac{\delta}{\sigma^*} \right) - 1,
%	\end{equation}
%	where $ \sigma^* = \sqrt{ \frac{\left\|\mathbf{1}_{1 \times p^{*}} \mathcal{Q}_{B^{*}}\right\|^{2}}{n^{2}} \sum_{i=1}^{n} \varphi^{2(i-1)}\left\|\mathcal{P}_{i}^2\right\|^{2}}$ and
%	$|\mu| \leq \frac{1}{n}\left\|\mathbf{1}_{1 \times p^{*}} \mathcal{Q}_{B^{*}}\right\|\left\|\sum_{k=0}^{\infty} r(k) + \sum_{i=1}^{n}\varphi^{i-1} \mathcal{P}_i^1v_1(i-1) \right\| $ 
%	 and $ \Phi  $ is the cumulative distribution function of the standard normal distribution,
	if the following statements hold:
	\begin{enumerate}
		\item no alarm is triggered;
		\item $ 	
		{\rm rank} \begin{bmatrix}
		\mathcal{O}_{n-1} & \mathcal{J}_{n-1}
		\end{bmatrix} - {\rm rank} [ \mathcal{J}_{n-1}]
		= n
		$.
	\end{enumerate}
\end{theorem}

\section{Numerical Examples}
Consider the following network composed of $ 4 $ agents:
\begin{figure}[!htp]
	\centering
	\begin{tikzpicture}
	[scale=1 ,auto=left,every node/.style={circle,fill=blue!30}]
	\node (n3) at (1,1) {3};
	\node (n4) at (-1,1) {4};
	\node (n1) at (-1,-1) {1};
	\node (n2) at (1,-1) {2};
	\foreach \from/\to in {n1/n2,n1/n4,n2/n3,n3/n4}
	\draw (\from) -- (\to);
	\end{tikzpicture}
	\caption{Network Topology}
\end{figure}
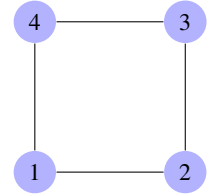
%		\begin{figure}[!htp]
%		\centering
%		\begin{tikzpicture}
%		[scale=1.2 ,auto=left,every node/.style={circle,fill=blue!20}]
%		\node (n3) at (0.8,0) {3};
%		\node (n4) at (-0.8,0) {4};
%		\node (n1) at (-2,-1) {1};
%		\node (n2) at (2,-1) {2};
%		\node (n5) at (-2,1) {5};
%		\node (n6) at (2,1) {6};
%		\foreach \from/\to in {n1/n2,n1/n4,n1/n5,n2/n3,n2/n6,n3/n4,n3/n6,n4/n5,n5/n6}
%		\draw (\from) -- (\to);
%		\end{tikzpicture}
%		\caption{Network Topology}
%	\end{figure}

Suppose that the weight matrix  is
\[A= 
\begin{bmatrix}
0.136 & 0.461 & 0 & 0.403  \\
0.461 & 0.153 & 0.386 & 0  \\
0 & 0.386 & 0.278 & 0.336  \\
0.403 & 0 & 0.336 & 0.261  
\end{bmatrix},
\]
which is generated randomly. Suppose that the initial state values of agents are $	x(0) = \begin{bmatrix}
100 & -50 & 50 & -100
\end{bmatrix}^\top.  $
Without loss of generality, assume that agent $ 1 $ is benign and it is running an attack detector.
Then the matrix $ C $ is 
$
\begin{bmatrix}
1 & 0 & 0 & 0  \\
0 & 1 & 0 & 0  \\
0 & 0 & 0 & 1 
\end{bmatrix}.
$

Suppose that agent $ 3 $ is both malicious and curious and other agents all are benign. Since for any agent $ j $, $ j = 1,2,4 $, $ \mathcal{N}_i \cup \{i\} \nsubseteq \mathcal{N}_3 \cup \{3\}  $, according to Theorem \ref{privacy}, the initial state privacy of every benign agent is guaranteed.
Since $ {\rm rank} \begin{bmatrix}
\mathcal{O}_{3} & \mathcal{J}_{3}
\end{bmatrix} - {\rm rank} [ \mathcal{J}_{3}]
= 4 $, according to Theorem \ref{Detect}, there is no undetectable input. In order to avoid being detected, agent $ 3 $ inputs the attack signals $ u_3(k) = -24\times 0.2^{k} $ at every time step $ k $ into the system.

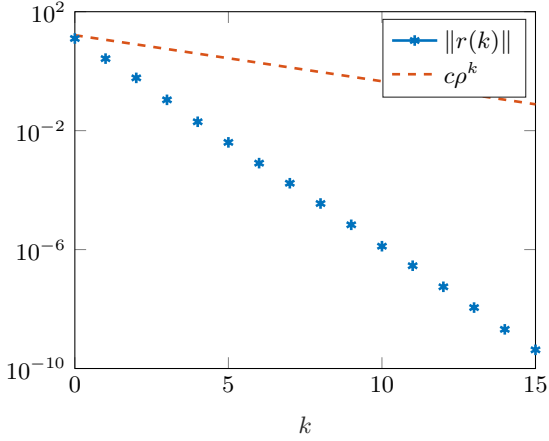
\begin{figure}[htbp] 
	\centering
\definecolor{mycolor1}{rgb}{0.00000,0.44700,0.74100}%
\definecolor{mycolor2}{rgb}{0.85000,0.32500,0.09800}%
\definecolor{mycolor3}{rgb}{0.92900,0.69400,0.12500}%
\definecolor{mycolor4}{rgb}{0.49400,0.18400,0.55600}%
\definecolor{mycolor5}{rgb}{0.46600,0.67400,0.18800}%
\definecolor{mycolor6}{rgb}{0.30100,0.74500,0.93300}%
\definecolor{mycolor7}{rgb}{0.63500,0.07800,0.18400}%
\begin{tikzpicture}

\begin{axis}[%
width=2.411in,
height=1.869in,
at={(1.011in,0.642in)},
scale only axis,
xmin=0,
xmax=15,
xlabel style={font=\color{white!15!black}},
xlabel={$k$},
ymode=log,
ymin=1e-10,
ymax=100,
yminorticks=true,
axis background/.style={fill=white},
legend style={legend cell align=left, align=left, draw=white!15!black}
]
\addplot [color=mycolor1, line width=1.0pt, draw=none, mark=asterisk, mark options={solid, mycolor1}]
  table[row sep=crcr]{%
0	12.565485098184\\
1	2.64354929493693\\
2	0.594970834952994\\
3	0.108531167438114\\
4	0.0198758209032058\\
5	0.00403163462498659\\
6	0.000807399734492917\\
7	0.000170337008785995\\
8	3.54623928306833e-05\\
9	6.80678862115301e-06\\
10	1.29895104616161e-06\\
11	2.86329942993965e-07\\
12	5.6015857282441e-08\\
13	1.12050695585097e-08\\
14	2.07761115739801e-09\\
15	4.24865081647835e-10\\
};
\addlegendentry{$\|r(k)\|$}

\addplot [color=mycolor2, dashed, line width=1.0pt]
  table[row sep=crcr]{%
0	16.2\\
1	11.34\\
2	7.938\\
3	5.5566\\
4	3.88962\\
5	2.722734\\
6	1.9059138\\
7	1.33413966\\
8	0.933897762\\
9	0.6537284334\\
10	0.45760990338\\
11	0.320326932366\\
12	0.2242288526562\\
13	0.15696019685934\\
14	0.109872137801538\\
15	0.0769104964610765\\
};
\addlegendentry{$c\rho^k$}

\addplot [color=mycolor3, dashed, line width=1.0pt, forget plot]
  table[row sep=crcr]{%
0	0\\
1	0\\
2	0\\
3	0\\
4	0\\
5	0\\
6	0\\
7	0\\
8	0\\
9	0\\
10	0\\
11	0\\
12	0\\
13	0\\
14	0\\
15	0\\
};
\addplot [color=mycolor4, dashed, line width=1.0pt, forget plot]
  table[row sep=crcr]{%
0	0\\
1	0\\
2	0\\
3	0\\
4	0\\
5	0\\
6	0\\
7	0\\
8	0\\
9	0\\
10	0\\
11	0\\
12	0\\
13	0\\
14	0\\
15	0\\
};
\addplot [color=mycolor5, dashed, line width=1.0pt, forget plot]
  table[row sep=crcr]{%
0	0\\
1	0\\
2	0\\
3	0\\
4	0\\
5	0\\
6	0\\
7	0\\
8	0\\
9	0\\
10	0\\
11	0\\
12	0\\
13	0\\
14	0\\
15	0\\
};
\addplot [color=mycolor6, dashed, line width=1.0pt, forget plot]
  table[row sep=crcr]{%
0	0\\
1	0\\
2	0\\
3	0\\
4	0\\
5	0\\
6	0\\
7	0\\
8	0\\
9	0\\
10	0\\
11	0\\
12	0\\
13	0\\
14	0\\
15	0\\
};
\addplot [color=mycolor7, dashed, line width=1.0pt, forget plot]
  table[row sep=crcr]{%
0	0\\
1	0\\
2	0\\
3	0\\
4	0\\
5	0\\
6	0\\
7	0\\
8	0\\
9	0\\
10	0\\
11	0\\
12	0\\
13	0\\
14	0\\
15	0\\
};
\addplot [color=mycolor1, dashed, line width=1.0pt, forget plot]
  table[row sep=crcr]{%
0	0\\
1	0\\
2	0\\
3	0\\
4	0\\
5	0\\
6	0\\
7	0\\
8	0\\
9	0\\
10	0\\
11	0\\
12	0\\
13	0\\
14	0\\
15	0\\
};
\addplot [color=mycolor2, dashed, line width=1.0pt, forget plot]
  table[row sep=crcr]{%
0	0\\
1	0\\
2	0\\
3	0\\
4	0\\
5	0\\
6	0\\
7	0\\
8	0\\
9	0\\
10	0\\
11	0\\
12	0\\
13	0\\
14	0\\
15	0\\
};
\addplot [color=mycolor3, dashed, line width=1.0pt, forget plot]
  table[row sep=crcr]{%
0	0\\
1	0\\
2	0\\
3	0\\
4	0\\
5	0\\
6	0\\
7	0\\
8	0\\
9	0\\
10	0\\
11	0\\
12	0\\
13	0\\
14	0\\
15	0\\
};
\addplot [color=mycolor4, dashed, line width=1.0pt, forget plot]
  table[row sep=crcr]{%
0	0\\
1	0\\
2	0\\
3	0\\
4	0\\
5	0\\
6	0\\
7	0\\
8	0\\
9	0\\
10	0\\
11	0\\
12	0\\
13	0\\
14	0\\
15	0\\
};
\addplot [color=mycolor5, dashed, line width=1.0pt, forget plot]
  table[row sep=crcr]{%
0	0\\
1	0\\
2	0\\
3	0\\
4	0\\
5	0\\
6	0\\
7	0\\
8	0\\
9	0\\
10	0\\
11	0\\
12	0\\
13	0\\
14	0\\
15	0\\
};
\addplot [color=mycolor6, dashed, line width=1.0pt, forget plot]
  table[row sep=crcr]{%
0	0\\
1	0\\
2	0\\
3	0\\
4	0\\
5	0\\
6	0\\
7	0\\
8	0\\
9	0\\
10	0\\
11	0\\
12	0\\
13	0\\
14	0\\
15	0\\
};
\addplot [color=mycolor7, dashed, line width=1.0pt, forget plot]
  table[row sep=crcr]{%
0	0\\
1	0\\
2	0\\
3	0\\
4	0\\
5	0\\
6	0\\
7	0\\
8	0\\
9	0\\
10	0\\
11	0\\
12	0\\
13	0\\
14	0\\
15	0\\
};
\addplot [color=mycolor1, dashed, line width=1.0pt, forget plot]
  table[row sep=crcr]{%
0	0\\
1	0\\
2	0\\
3	0\\
4	0\\
5	0\\
6	0\\
7	0\\
8	0\\
9	0\\
10	0\\
11	0\\
12	0\\
13	0\\
14	0\\
15	0\\
};
\addplot [color=mycolor2, dashed, line width=1.0pt, forget plot]
  table[row sep=crcr]{%
0	0\\
1	0\\
2	0\\
3	0\\
4	0\\
5	0\\
6	0\\
7	0\\
8	0\\
9	0\\
10	0\\
11	0\\
12	0\\
13	0\\
14	0\\
15	0\\
};
\addplot [color=mycolor3, dashed, line width=1.0pt, forget plot]
  table[row sep=crcr]{%
0	0\\
1	0\\
2	0\\
3	0\\
4	0\\
5	0\\
6	0\\
7	0\\
8	0\\
9	0\\
10	0\\
11	0\\
12	0\\
13	0\\
14	0\\
15	0\\
};
\end{axis}
\end{tikzpicture}%
	\caption{One snapshot of the comparison between $ \|r(k)\| $ and $ c\rho^k $}
\end{figure}

\begin{figure}[htbp] 
	\centering
\definecolor{mycolor1}{rgb}{0.00000,0.44700,0.74100}%
\definecolor{mycolor2}{rgb}{0.85000,0.32500,0.09800}%
\definecolor{mycolor3}{rgb}{0.92900,0.69400,0.12500}%
\definecolor{mycolor4}{rgb}{0.49400,0.18400,0.55600}%
\begin{tikzpicture}

\begin{axis}[%
width=2.411in,
height=1.869in,
at={(1.011in,0.642in)},
scale only axis,
xmin=0,
xmax=15,
xlabel style={font=\color{white!15!black}},
xlabel={$k$},
ymin=-100,
ymax=100,
axis background/.style={fill=white},
legend style={legend cell align=left, align=left, draw=white!15!black}
]
\addplot [color=mycolor1, line width=1.2pt]
  table[row sep=crcr]{%
0	100\\
1	-49.1627108881014\\
2	31.6063821854784\\
3	-26.1911063894635\\
4	5.30031506845636\\
5	-14.8639574205641\\
6	-2.94088651737126\\
7	-10.2367969925849\\
8	-5.83637485578583\\
9	-8.50617918750737\\
10	-6.89024065440906\\
11	-7.86923628387165\\
12	-7.27634415374995\\
13	-7.63545860632765\\
14	-7.41795505001947\\
15	-7.54969240106956\\
};
\addlegendentry{agent 1}

\addplot [color=mycolor2, line width=1.2pt]
  table[row sep=crcr]{%
0	-50\\
1	57.3694122138814\\
2	-37.6126073312767\\
3	12.6219907350534\\
4	-19.330456659355\\
5	-0.267436097584382\\
6	-11.8664309464477\\
7	-4.85262309655477\\
8	-9.1029633030272\\
9	-6.52905425354576\\
10	-8.08806823631907\\
11	-7.14382084574057\\
12	-7.7157291456502\\
13	-7.36933824781865\\
14	-7.57913868433022\\
15	-7.45206764562004\\
};
\addlegendentry{agent 2}

\addplot [color=mycolor3, line width=1.2pt]
  table[row sep=crcr]{%
0	50\\
1	-62.6645116607626\\
2	9.92937815399196\\
3	-23.6509306993531\\
4	0.635314563280099\\
5	-12.8871336212229\\
6	-4.35661262656315\\
7	-9.4336091860681\\
8	-6.33611756858019\\
9	-8.20667049939144\\
10	-7.07239785521216\\
11	-7.7590847853801\\
12	-7.34310108586288\\
13	-7.59503501051476\\
14	-7.44244079940842\\
15	-7.53486247317544\\
};
\addlegendentry{agent 3}

\addplot [color=mycolor4, line width=1.2pt]
  table[row sep=crcr]{%
0	-100\\
1	30.7494404938928\\
2	-32.4894011496814\\
3	7.52100614348473\\
4	-16.5696987018475\\
5	-1.9747742003193\\
6	-10.8346570084889\\
7	-5.4767516084334\\
8	-8.72451974939831\\
9	-6.758086098152\\
10	-7.94929052166006\\
11	-7.22785751257328\\
12	-7.6648254664415\\
13	-7.40016810371659\\
14	-7.56046545979791\\
15	-7.46337747959624\\
};
\addlegendentry{agent 4}

\addplot [color=black, dashed, forget plot]
  table[row sep=crcr]{%
0	0\\
1	0\\
2	0\\
3	0\\
4	0\\
5	0\\
6	0\\
7	0\\
8	0\\
9	0\\
10	0\\
11	0\\
12	0\\
13	0\\
14	0\\
15	0\\
};
\node[right, align=left]
at (axis cs:11,-15) {-7.5000};
\end{axis}
\end{tikzpicture}%
	\caption{The trajectory of each state value $ x_i(k) $. The blue, red, yellow and purple lines correspond to $ x_1(k), x_2(k), x_3(k), x_4(k) $ respectively. The black dashed line corresponds to the average value of the initial state $ x(0) $. The number ``$-7.5000$'' above these lines corresponds to the final convergence value of asymptotic consensus.}
\end{figure}
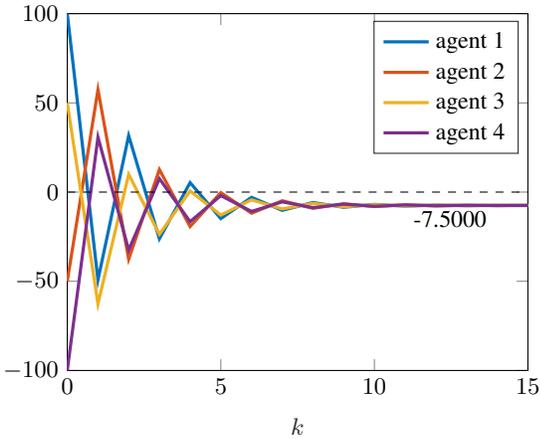

Suppose that agent $ 1 $ selects these parameters as follow: $ c = 16.2, \rho = 0.7, \varphi = 0.2$. According to Theorem \ref{An Estimation of False Alarm Rate}, false alarm rate $ \alpha $ is no more than $  0.01$. Since that no alarm is triggered after $ 2000 $ time steps have been run, and the state values of the neighbors of agent $ 1 $ and its own state value have always been $ -7.5000 $ since the $ 30 $-th time step, it can be considered that the system has achieved an asymptotic consensus. According to Theorem \ref{converage}, when an asymptotic consensus is achieved, the convergence rate $ \varrho \leq \max\{\rho,|\lambda_2|,|\lambda_n| \} = \max\{0.7, |0.2229|,|-0.6057|\} =  0.7$. One snapshots of the comparison between $ \|r(k)\| $ and $ c\rho^k $ and the trajectories of agents' state values are shown in Fig. 2, Fig 3, respectively. From Fig 3, it can be seen that although an asymptotic consensus is achieved, the final convergence value $ -7.5000 $ is not the average value $ 0 $ of the initial state $ x(0) $. 

Now, agent $ 1$ begin to estimate the error of the final convergence value. 
Since agent $ 1 $ does not know which agents are malicious attackers, according to Theorem \ref{error}, it needs to consider all cases that meet the detectable condition $ {\rm rank} \begin{bmatrix}
\mathcal{O}_{3} & \mathcal{J}_{3}
\end{bmatrix} - {\rm rank} [ \mathcal{J}_{3}]
= 4  $. 
%After verification, all possible of the collection of malicious attackers meeting the detectable condition above are $\{2\},\{3\},\{4\},\{2,3\},\{3,4\} .$ Among them, when the collection of malicious attackers is $\{2,3\}$,
%$\left\|\mathbf{1}_{1 \times p^{*}} \mathcal{Q}_{B^{*}}\right\| = 4.1466$. 
According to Theorem \ref{error}, if let $ \beta = 0.001 $, $ [-57.9926,57.9926] $
is a confidence interval for $ e $ with confidence coefficient of not less than $ 0.999 $. If agent $ 1 $ has known that there is at most one malicious attacker in the system, $ [-29.5478,29.5478] $
is a confidence interval for $ e $ with confidence coefficient of not less than $ 0.999 $.
%for any $ \delta > 0 $, $ \mathbb{P}\left[\left|e -  \mu \right| \leq \delta \right] \geq 2\Phi\left(\frac{\delta}{ 1.0586} \right) - 1 $, where $ |\mu| \leq 1.4316 $. Thus, we can get some estimates of error $e$ as shown in Table \ref{tab1}.
%The actual error $ e = 1.1029 $.
%	\begin{table}[!htb]
%		\centering
%		\caption{Some estimates about error $ e $, where $ |\mu| \leq 1.4316 $.}
%		\begin{tabular}{c|c}
%			\hhline
%			$ \delta =  2.7268  $        & $ \mathbb{P}\left[\left|e -  \mu \right| \leq 2.7268  \right] \geq 99\% $ \\ \hline
%			$ \delta = 3.4833 $         & $ \mathbb{P}\left[\left|e -   \mu \right| \leq 3.4833 \right] \geq 99.9\% $ \\ \hline
%		  $ \delta = 4.1186 $         & $ \mathbb{P}\left[\left|e -   \mu \right| \leq 4.1186 \right] \geq 99.99\% $ \\ \hline
%			\hhline
%		\end{tabular}
%	\label{tab1}
%	\end{table}

\section{Conclusion}

In this paper, we deal with the case that the consensus system is threatened by a set of unknown agents that are both ``malicious" and ``curious". We propose a privacy-preserving average consensus algorithm equipped with an attack detector with a time-varying exponentially decreasing threshold, for every benign agent, which can guarantee the initial state privacy of every benign agent, under mild conditions.
An upper bound of false alarm rate and the necessary and sufficient condition for there is no undetectable input by the attack detector
in the system are given. We prove that the system can achieve asymptotic consensus almost surely and give an upper bound of convergence rate and
some estimates about the error.

\end{document}